\documentclass{IEEEtran}

\usepackage[noadjust]{cite}
\usepackage{CJKutf8, amsmath, amssymb, enumitem, orcidlink}

\usepackage{tikz}
\usetikzlibrary{fit,arrows.meta,decorations.pathreplacing}

\newtheorem{lemma}{Lemma}
\newtheorem{theorem}{Theorem}
\newtheorem{corollary}{Corollary}

\newtheorem{definition}{Definition}

\DeclareMathOperator{\diag}{diag}
\DeclareMathOperator{\poly}{poly}
\DeclareMathOperator{\tr}{tr}
\DeclareMathOperator*{\E}{\mathbb E}

\begin{document}

\bstctlcite{BSTcontrolNoURL}

\title{Approximating local properties by tensor network states with constant bond dimension}

\author{Yichen Huang (黄溢辰)\orcidlink{0000-0002-8496-9251}
\thanks{A related paper by Dalzell and Brand\~ao \cite{DB19} appeared on arXiv simultaneously. I wrote a Perspective \cite{Hua19} on Ref.~\cite{DB19}. Both their paper and the present work were accepted as talks and merged into a single talk given by Dalzell (video DOI: \href{https://doi.org/10.12351/ks.2001.0304}{10.12351/ks.2001.0304}) at the 23rd Annual Conference on Quantum Information Processing (QIP 2020).}
\thanks{The author was with Microsoft Research AI, Redmond, WA 98052 USA (e-mail: \href{mailto:huangtbcmh@gmail.com}{huangtbcmh@gmail.com}).}
}

\begin{CJK}{UTF8}{gbsn}

\maketitle

\end{CJK}

\begin{abstract}

Classical simulation of quantum many-body systems is a fundamental challenge due to their exponentially large Hilbert spaces. Tensor network states are a powerful ansatz to efficiently represent many physically relevant quantum states. A key question is the bond dimension---which determines the number of parameters in the ansatz---required to approximate all local properties to accuracy $\delta$. In one dimension, we prove that an area law for the R\'enyi entanglement entropy $R_\alpha$ with index $\alpha<1$ implies a matrix product state representation with bond dimension $\poly(1/\delta)$. For (at most constant-fold degenerate) ground states of one-dimensional gapped Hamiltonians, a bond dimension almost linear in $1/\delta$ suffices. In two dimensions, an area law for $R_\alpha(\alpha<1)$ implies a projected entangled pair state representation with bond dimension $e^{O(1/\delta)}$. In both one and two dimensions, analogous results are obtained for states with logarithmic corrections to the area law. These findings rigorously justify the common practice of using a system-size-independent bond dimension in tensor network simulations.

\end{abstract}


\begin{figure*}[h]
  \centering
  \begin{tikzpicture}[
    site/.style={circle, draw=black, fill=white, minimum size=8, inner sep=0},
    tensor/.style={rectangle, draw=black, fill=black!10, minimum size=8, inner sep=0},
    bond/.style={line width=1.2},
    leg/.style={line width=0.8},
    highlight/.style={draw=black!70, thick, dashed, rounded corners, inner sep=4},
    >=Stealth
  ]

    \begin{scope}[shift={(-5,0)}]
      \node at (-2.7,0) {\dots};
      \foreach \i in {1,...,8} {
        \node[site] (L\i) at (0.6*\i-2.7,0) {};
      }
      \node at (2.7,0) {\dots};
      \node at (0,0.8) {1D spin chain: state $|\psi\rangle$};
      \node at (0,-0.8) {local region};
      \draw[->] (0,-0.6) -- (0,-0.3);
      \node[highlight, fit=(L4)(L5)] {};
      \node at (0,-1.6) {For $|\psi\rangle$ satisfying R\'enyi area law:};
    \end{scope}
    
    \node[font=\large] (arrow) at (0,0) {$\longrightarrow$};
    \node[above] {approximated by};

    \begin{scope}[shift={(5,0)}]
      \node at (-2.7,0) {\dots};
      \foreach \i in {1,...,8} {
        \node[site] (R\i) at (0.6*\i-2.7,0) {};
      }
      \node at (2.7,0) {\dots};
      \node at (-2.7,-0.8) {\dots};
      \foreach \i in {1,...,8} {
        \node[tensor] (T\i) at (0.6*\i-2.7,-0.8) {};
        \draw[leg] (R\i) -- (T\i);
      }
      \node at (2.7,-0.8) {\dots};
      \draw[bond] (-2.4,-0.8) -- (T1);
      \foreach \i in {1,...,7} {
        \pgfmathtruncatemacro\j{\i+1}
        \draw[bond] (T\i) -- (T\j);
      }
      \draw[bond] (T8) -- (2.4,-0.8);
      \node[highlight, fit=(R4)(R5)] {};
      \node at (0,0.8) {matrix product state $|\phi\rangle$};
      \node at (0,-1.6) {bond dimension $D$};
      \draw[->] (0,-1.4) -- (0,-0.85);
    \end{scope}

    \node at (0,-2.4) {system-size-independent $D=\poly(1/\delta)$ ensures $\big|\langle\psi|\hat O|\psi\rangle-\langle\phi|\hat O|\phi\rangle\big|\le\delta,~\forall$ local observable $\hat O$};
  \end{tikzpicture}
  \caption{Constant-bond-dimension tensor networks for local properties (1D matrix-product-state illustration). In a spin chain, a state $|\psi\rangle$ (left) satisfying an area law for the R\'enyi entanglement entropy with index $\alpha<1$ can be approximated by a matrix product state $|\phi\rangle$ (right) with a system-size-independent bond dimension $D=\poly(1/\delta)$, such that all local properties are accurate up to error $\delta$ (Theorem \ref{t1}). Analogous results are also obtained for other classes of states: ground states of 1D gapped Hamiltonians (Theorem \ref{t2}), area-law states in 2D (Theorem \ref{2d}), and states with logarithmic corrections to the area law in both 1D (Corollary \ref{c2}) and 2D (Corollary \ref{c1}).}
\end{figure*}

\section{Introduction}

Classical simulation of quantum many-body systems is a fundamental pursuit in modern physics, yet it presents a formidable challenge due to the ``curse of dimensionality''---the exponential growth of the Hilbert space dimension with system size. Thus, a complete description of generic quantum states is computationally prohibitive except for relatively small systems. Fortunately, many states of significant physical interest, such as ground states of local Hamiltonians, are non-generic. They often exhibit substantial structure that confines them to a small, manageable corner of the vast Hilbert space. Tensor network states, notably matrix product states (MPS) \cite{PVWC07, FNW92} for one spatial dimension and projected entangled pair states (PEPS) \cite{VC04} for two and higher dimensions, are powerful, versatile, and widely adopted ans\"atze for efficiently parameterizing these structured states. A crucial aspect of a tensor network representation is its ``bond dimension,'' which dictates the number of parameters in the ansatz and is a primary factor governing the required computational resources---memory and runtime. Therefore, a fundamental and practically vital question is to determine the bond dimension that guarantees an accurate approximation to the physical properties of interest, especially local observables.

As a variational algorithm over MPS \cite{Sch11}, the celebrated density matrix renormalization group (DMRG) \cite{Whi92, Whi93} has established itself as the leading numerical method for simulating one-dimensional (1D) quantum systems. Complementing this practical success, significant progress has been made in establishing the theoretical foundations of DMRG. For ground states of 1D gapped Hamiltonians, notable rigorous results include
\begin{itemize} [noitemsep]
\item area laws for the entanglement entropy \cite{Has07, AKLV13, Hua14a};
\item efficient MPS approximations to the wave function \cite{Has07, AKLV13, Hua14a}, i.e., there exists an MPS with bond dimension less than a polynomial in the system size such that the fidelity approaches $1$;
\item efficient algorithms for finding such MPS approximations \cite{LVV15, Hua14b, CF16, Hua15c, ALVV17}.
\end{itemize}
Crucially, an area law for the von Neumann entanglement entropy in 1D does not necessarily imply efficient MPS approximations \cite{SWVC08}.

In practice, accurate results can often be obtained using MPS with quite small bond dimension. Extreme examples are the infinite DMRG \cite{McC08} and infinite (imaginary-)time-evolving block decimation \cite{Vid07} algorithms, which work directly in the thermodynamic limit. It was observed that a constant bond dimension is sufficient for computing local expectation values. This observation cannot be explained by Refs.~\cite{Has07, AKLV13, Hua14a}, for the upper bounds on the bond dimension proved there become infinite in the thermodynamic limit.

It is important to rigorously justify the empirical success of using a constant bond dimension. To this end, we prove that
\begin{enumerate} [noitemsep]
\item In 1D, an area law for the R\'enyi entanglement entropy $R_\alpha$ with index $\alpha<1$ implies an MPS representation with bond dimension $\poly(1/\delta)$ such that all local properties are approximated to additive error $\le\delta$.
\item For (at most constant-fold degenerate) ground states of 1D gapped Hamiltonians, a bond dimension almost linear in $1/\delta$ suffices.
\end{enumerate}

Analogous results with even better upper bounds on the bond dimension were obtained \cite{Hua15a, SV17} for (positive semidefinite) matrix product operator (MPO) \cite{VGC04, ZV04} approximations (MPO are a generalization of MPS to operators). However, MPS are preferable to MPO, for the latter are more difficult to work with in both theory \cite{KGE14} and practice.

In 2D, rigorous results are scarce. In particular, it is an open problem whether ground states of gapped Hamiltonians always satisfy an area law \cite{ECP10, Osb12, GHLS15, Hua15b}, and an area law (in its strongest formulation) does not imply efficient tensor network approximations to the wave function \cite{GE16, Hua20ISIT, Hua21JCP}. PEPS are a generalization of MPS to higher dimensions. We prove that
\begin{enumerate} [resume*]
\item In 2D, an area law for $R_\alpha(\alpha<1)$ implies a PEPS representation with bond dimension $e^{O(1/\delta)}$ such that all local properties are approximated to additive error $\le\delta$.
\end{enumerate}

Analogous results for approximating local properties are obtained for states (e.g., many critical ground states \cite{HLW94, VLRK03, LRV04, CC04, CC09, Wol06, GK06}) with logarithmic corrections to the area law: 
\begin{enumerate} [resume*]
\item In 1D, there exists an MPS approximation with bond dimension $\poly(1/\delta)$. 
\item In 2D, there exists a PEPS approximation with bond dimension $e^{\tilde O(1/\delta)}$.\footnote{To simplify the notation, we use a tilde to hide a polylogarithmic factor, e.g., $\tilde O(x):=O(x\poly\log x)$.}
\end{enumerate}

Our results establish the existence of tensor network representations with the specified bond dimensions for a given local approximation guarantee in various physical scenarios. The computation of local expectation values, however, is a distinct problem. In particular, it is not well understood under what conditions local properties of ground states can be rigorously computed in time independent of the system size \cite{Hua20}.

\section{Results in one dimension}

In this and the next sections, we restrict ourselves to 1D with open boundary conditions. Consider a chain of $N$ qu\emph{d}its (or spin-$\frac{d-1}{2}$'s) with $d=\Theta(1)$. Let $\mathcal H_i=\mathbb C^d$ be the Hilbert space of qudit $i$, and define
\begin{equation}
\mathcal H_{[i,j]}=\bigotimes_{k=\max\{i,1\}}^{\min\{j,N\}}\mathcal H_k
\end{equation}
as the Hilbert space of qudits with indices in the interval $[i,j]$.

\begin{definition} [matrix product states \cite{PVWC07, FNW92}]
Let $\{|j_i\rangle\}_{j_i=1}^d$ be the computational basis of $\mathcal H_i$, and let $\{D_i\}_{i=0}^N$ be a sequence of positive integers with $D_0=D_N=1$. An MPS $|\phi\rangle$ has the form
\begin{equation} \label{mps}
|\phi\rangle=\sum_{j_1,j_2,\ldots,j_N=1}^dA_{j_1}^{[1]}A_{j_2}^{[2]}\cdots A_{j_N}^{[N]}|j_1j_2\cdots j_N\rangle,
\end{equation}
where $A_{j_i}^{[i]}$ is a matrix of size $D_{i-1}\times D_i$. Terminology: $D_i$ is called the bond dimension across the cut $\mathcal H_{[1,i]}\otimes\mathcal H_{[i+1,N]}$, and $D:=\max_{0\le i\le N}D_i$ is the bond dimension of the MPS $|\phi\rangle$.
\end{definition}

Any state can be expressed exactly as an MPS with exponential bond dimension $D\le d^{\lfloor N/2\rfloor}$, where $\lfloor\cdot\rfloor$ denotes the floor function.

\begin{definition} [entanglement entropy]
The R\'enyi entanglement entropy $R_\alpha$ with index $\alpha\in(0,1)\cup(1,+\infty)$ of a pure state $|\psi\rangle$ between a subsystem $X$ and its complement $\overline X$ is defined as
\begin{equation}
R_\alpha(\rho_X)=\frac{1}{1-\alpha}\ln\tr(\rho_X^\alpha),
\end{equation}
where $\rho_X=\tr_{\overline X}(|\psi\rangle\langle\psi|)$ is the reduced density matrix. The von Neumann entanglement entropy is
\begin{equation}
S(\rho_X)=-\tr(\rho_X\ln\rho_X)=\lim_{\alpha\to1}R_\alpha(\rho_X).
\end{equation}
\end{definition}

It is well known (and not difficult to prove) that $R_\alpha$ is monotonically non-increasing with respect to $\alpha$, i.e., $R_\alpha(\rho_X)\ge R_\beta(\rho_X)$ if $\alpha\le\beta$.

\begin{theorem} \label{t1}
Suppose that $|\psi\rangle$ satisfies an area law for the R\'enyi entanglement entropy $R_\alpha$ with $\alpha<1$, in the sense that the entropy across any cut is $O(1)$. Then for any $\delta>0$, there exists an MPS $|\phi\rangle$ with bond dimension
\begin{equation}
D=\tilde O\left(\delta^{-1-\frac{3\alpha}{1-\alpha}}\right)
\end{equation}
such that $|\langle\psi|\hat O|\psi\rangle-\langle\phi|\hat O|\phi\rangle|\le\delta$ for any local observable $\hat O$ with $\|\hat O\|\le1$.
\end{theorem}

\begin{theorem} \label{t2}
Suppose that $|\psi\rangle$ is the (at most constant-fold degenerate) ground state of a 1D Hamiltonian with a constant energy gap $\Delta$. Then for any $\delta>0$, there exists an MPS $|\phi\rangle$ with bond dimension
\begin{equation}
D=e^{\tilde O\left(\sqrt[4]{\frac{1}{\Delta}\log^3\frac{1}{\delta}}\right)}/\delta=O(\delta^{-1-\gamma}),\quad\forall\gamma>0
\end{equation}
such that $|\langle\psi|\hat O|\psi\rangle-\langle\phi|\hat O|\phi\rangle|\le\delta$ for any local observable $\hat O$ with $\|\hat O\|\le1$.
\end{theorem}

It is not necessary to assume exact ground-state degeneracy. The degeneracy condition enters the proof of Theorem \ref{t2} solely via Lemma \ref{gap} (below), which, in particular, remains valid in the presence of an exponentially small $e^{-\Omega(N)}$ splitting of the degeneracy \cite{Hua14a} (as is typically observed in physical systems).

\section{Proofs for one dimension}

\subsection{Preliminaries}

\begin{lemma} [canonical form \cite{PVWC07}]
Any MPS can be transformed into the so-called canonical form without increasing the bond dimension across any cut such that
\begin{gather}
\sum_{j=1}^dA_j^{[i]}A_j^{[i]\dag}=I,\quad\sum_{j=1}^{d}A_j^{[i]\dag}\Lambda^{[i-1]}A_j^{[i]}=\Lambda^{[i]},\\
\Lambda^{[i]}=\diag\left\{\left(\lambda_1^{[i]}\right)^2,\left(\lambda_2^{[i]}\right)^2,\ldots\right\},
\end{gather}
where $\lambda_1^{[i]}\ge\lambda_2^{[i]}\ge\cdots>0$ are the Schmidt coefficients of the MPS across the cut $\mathcal H_{[1,i]}\otimes\mathcal H_{[i+1,N]}$ in non-ascending order.
\end{lemma}

Fix a cut, and let $\lambda_1\ge\lambda_2\ge\cdots>0$ be the Schmidt coefficients of $|\psi\rangle$ across the cut, where $\sum_j\lambda_j^2=1$. The ``truncation error'' can be upper bounded as follows.

\begin{lemma} [\cite{VC06}] \label{area}
Let $R_\alpha(\alpha<1)$ be the R\'enyi entanglement entropy of $|\psi\rangle$ across the cut. Then,
\begin{equation}
\sum_{j>D'}\lambda_j^2\le e^{\frac{1-\alpha}{\alpha}(R_\alpha-\ln D')}.
\end{equation}
\end{lemma}

\begin{lemma} [\cite{Hua14a}] \label{gap}
Suppose that $|\psi\rangle$ is the (at most constant-fold degenerate) ground state of a 1D Hamiltonian $H=\sum_{i=1}^{N-1}H_i$ with a constant energy gap $\Delta=\Theta(1)$, where $H_i$ acts on $\mathcal H_{[i,i+1]}$ and satisfies $\|H_i\|\le1$. Then, $\sum_{j>D'}\lambda_j^2\le\epsilon$ for
\begin{equation}
D'=e^{\tilde O\left(\frac{1}{\Delta}+\sqrt[4]{\frac{1}{\Delta}\log^3\frac{1}{\epsilon}}\right)}.
\end{equation}
\end{lemma}

\subsection{Main lemma}

\begin{lemma} \label{trunc}
For any integer $D'\ge2$, let $\epsilon_{D'}$ satisfy
\begin{equation}
\epsilon_{D'}\ge\max_{1\le i\le N-1}\sum_{j>D'}\left(\lambda_j^{[i]}\right)^2,
\end{equation}
where $\lambda_1^{[i]}\ge\lambda_2^{[i]}\ge\cdots>0$ are the Schmidt coefficients of $|\psi\rangle$ across the cut $\mathcal H_{[1,i]}\otimes\mathcal H_{[i+1,N]}$. Suppose that
\begin{equation} \label{eq:N}
N\ge\frac{C}{\sqrt[3]{\epsilon_{D'}\log D'}}\log\frac{D'}{\epsilon_{D'}}
\end{equation}
for a sufficiently large constant $C>0$. Then, there exists an MPS $|\phi\rangle$ with bond dimension
\begin{equation}
D=O\left(D'\sqrt[3]{\frac{\log^2D'}{\epsilon_{D'}}}\right)
\end{equation}
such that
\begin{equation} \label{eq:l4}
\left|\langle\psi|\hat O|\psi\rangle-\langle\phi|\hat O|\phi\rangle\right|\le O\left(\sqrt[3]{\epsilon_{D'}\log D'}\right)
\end{equation}
for any local observable $\hat O$ with $\|\hat O\|\le1$.
\end{lemma}

\begin{IEEEproof}
Schuch and Verstraete \cite{SV17} proved an analogous result, albeit restricted to MPO approximations. We build upon their approach, introducing additional technical ingredients to handle the MPS case. First, we express $|\psi\rangle$ exactly as an MPS (\ref{mps}) in canonical form with bond dimension $\le d^{\lfloor N/2\rfloor}$. Let $P$ be the diagonal matrix with entries $P_{jj}=1$ for $j\le D'$ and $P_{jj}=0$ otherwise. To simplify the notation, let $b:=\lceil\log_dD'\rceil$ and $m:=\lceil\sqrt[3]{b^2/\epsilon_{D'}}\rceil$, where $\lceil\cdot\rceil$ denotes the ceiling function. We may assume that $\epsilon_{D'}\log D'\le1$ so that $b=O(m)$, since otherwise (\ref{eq:l4}) trivially holds. For $i=1-m,2-m,\ldots,N-1$, let
\begin{multline}
|u'_i\rangle=\sum_{j_1,j_2,\ldots,j_N=1}^d\underbrace{A_{j_1}^{[1]}A_{j_2}^{[2]}\cdots A_{j_{i}}^{[i]}}_{\mathrm{without}~P}\\
\cdot\underbrace{PA_{j_{i+1}}^{[i+1]}PA_{j_{i+2}}^{[i+2]}P\cdots PA_{j_{i+m}}^{[i+m]}P}_{\mathrm{with}~P}\\
\cdot\underbrace{A_{j_{i+m+1}}^{[i+m+1]}A_{j_{i+m+2}}^{[i+m+2]}\cdots A_{j_N}^{[N]}}_{\mathrm{without}~P}|j_1j_2\cdots j_N\rangle
\end{multline}
be the unnormalized state obtained by truncating every bond from site $\max\{i,1\}$ to site $\min\{i+m+1,N\}$. A minor modification of the proof of Lemma 1 in Ref.~\cite{VC06} yields
\begin{equation}
\||\psi\rangle-|u'_i\rangle\|^2\le2\sum_{k=\max\{i,1\}}^{\min\{i+m,N-1\}}\sum_{j>D'}\left(\lambda_j^{[k]}\right)^2\le2(m+1)\epsilon_{D'}.
\end{equation}
Let $|u_i\rangle = |u'_i\rangle / \||u'_i\rangle\|$ be the normalized state. Consequently,
\begin{equation} \label{eq:u_app}
\left|\langle\psi|\hat O|\psi\rangle-\langle u_i|\hat O|u_i\rangle\right|\le O(\sqrt{m\epsilon_{D'}})
\end{equation}
for any observable $\hat O$ with $\|\hat O\|\le1$.

Since the Schmidt ranks of $|u_i\rangle$ across the cuts $\mathcal H_{[1,i]}\otimes\mathcal H_{[i+1,N]}$ and $\mathcal H_{[1,i+m]}\otimes\mathcal H_{[i+m+1,N]}$ are $\le D'$ by construction, there exist isometries
\begin{gather}
U_i^\mathrm{left}:\mathcal H_{[i-b+1,i]}\to\mathcal H_{[1,i]},\\
U_i^\mathrm{right}:\mathcal H_{[i+m+1,i+m+b]}\to\mathcal H_{[i+m+1,N]}
\end{gather}
such that
\begin{equation}
|u_i\rangle=U_i^\mathrm{left}\otimes I_{[i+1,i+m]}\otimes U_i^\mathrm{right}|v_i\rangle
\end{equation}
for some state $|v_i\rangle\in\mathcal H_{[i-b+1,i+m+b]}$, where $I_X$ is the identity operator on $\mathcal H_X$. Let $M:=m+2b$. For $i=1-m,2-m,\ldots,2b$, define
\begin{equation}
|w_i\rangle=\bigotimes_{k=0}^{\left\lfloor\frac{N-1-i}{M}\right\rfloor}(V_{i+kM}\otimes I_{[i+1+kM,i+m+b+kM]})|v_{i+kM}\rangle,
\end{equation}
where $V_{i+kM}$ is some unitary on $\mathcal H_{[i-b+1+kM,i+kM]}$. We view $|w_i\rangle$ as a state of the entire system $\mathcal H_{[1,N]}$ by implicitly tensoring it with the product state $|0\cdots0\rangle$ on the remaining qudits. For any observable $\hat O$ on $\mathcal H_{[i+1+kM,i+m+kM]}$,
\begin{equation} \label{error1}
\langle w_i|\hat O|w_i\rangle=\langle v_{i+kM}|\hat O|v_{i+kM}\rangle=\langle u_{i+kM}|\hat O|u_{i+kM}\rangle.
\end{equation}
Combining this with (\ref{eq:u_app}), we obtain
\begin{equation}
\left|\langle\psi|\hat O|\psi\rangle-\langle w_i|\hat O|w_i\rangle\right|\le O(\sqrt{m\epsilon_{D'}}).
\end{equation}

Suppose that $i\neq j$. If every $V_{i+kM}$ were a Haar-random unitary, then
\begin{multline}
\sup_{\|\hat O\|\le1}\E_{\{V_{i+kM}\}}\left|\langle w_i|\hat O|w_j\rangle\right|^2\le\sup_{\||W\rangle\|\le1}\E_{\{V_{i+kM}\}}|\langle w_i|W\rangle|^2\\
=\sup_{\||W\rangle\|\le1}\langle W|\left(\E_{\{V_{i+kM}\}}|w_i\rangle\langle w_i|\right)|W\rangle\le e^{-\frac{\Omega(Nb)}{M}}=:t.
\end{multline}
Markov's inequality implies that
\begin{equation}
\Pr\left(\left|\langle w_i|\hat O|w_j\rangle\right|^2\ge\sqrt t\right)\le\sqrt t
\end{equation}
for any fixed observable $\hat O$ with $\|\hat O\|\le1$. Since $j$ can take $<M$ values and the total number of linearly independent local operators is $O(N)$, the union bound over the $O(MN)$ events, together with (\ref{eq:N}), guarantees the existence of a specific set of unitaries $\{V_{i+kM}\}$ such that
\begin{equation}
\left|\langle w_i|\hat O|w_j\rangle\right|\le O\left(\sqrt[4]t\right)
\end{equation}
for all local operators $\hat O$ with $\|\hat O\|\le1$.

Define
\begin{equation}
|\phi\rangle=\frac{1}{\sqrt M}\sum_{i=1-m}^{2b}|w_i\rangle
\end{equation}
so that $\langle\phi|\hat O|\phi\rangle=T_\mathrm{diag}+T_\mathrm{off}$, where
\begin{gather}
T_\mathrm{diag}:=\frac{1}{M}\sum_{i=1-m}^{2b}\langle w_i|\hat O|w_i\rangle,\\
T_\mathrm{off}:=\frac{1}{M}\sum_{i=1-m}^{2b}\sum_{j\in\{1-m,2-m,\ldots,2b\}\setminus\{i\}}\langle w_i|\hat O|w_j\rangle.
\end{gather}
$T_\mathrm{diag}$ is the average of $M$ terms, $M-O(b)$ of which are close to $\langle\psi|\hat O|\psi\rangle$ in the sense of (\ref{error1}). $|T_\mathrm{off}|=O(M\sqrt[4]t)$ is negligible due to (\ref{eq:N}). Thus,
\begin{multline} \label{error3}
\left|\langle\psi|\hat O|\psi\rangle-\langle\phi|\hat O|\phi\rangle\right|\le O(\sqrt{m\epsilon_{D'}})+O(b/M)\\
\le O\left(\sqrt[3]{b\epsilon_{D'}}\right)=O\left(\sqrt[3]{\epsilon_{D'}\log D'}\right)
\end{multline}
for any local observable $\hat O$ with $\|\hat O\|\le1$. By construction, each $|w_i\rangle$ is an MPS with bond dimension $\le D'$. The bond dimension of $|\phi\rangle$ is
\begin{equation}
D\le D'M=D'(m+2b)=O\left(D'\sqrt[3]{\frac{\log^2D'}{\epsilon_{D'}}}\right).
\end{equation}
\end{IEEEproof}

\subsection{Proofs of main theorems}

The proofs of Theorems \ref{t1} and \ref{t2} assume (\ref{eq:N}) so that Lemma \ref{trunc} applies. Otherwise, the theorems follow by combining Lemma 1 in Ref.~\cite{VC06} with Lemmas \ref{area} and \ref{gap}, respectively.

\begin{IEEEproof} [Proof of Theorem \ref{t1}]
Lemma \ref{area} yields
\begin{equation}
\epsilon_{D'}=O\left(D'^{-\frac{1-\alpha}{\alpha}}\right).
\end{equation}
Lemma \ref{trunc} gives
\begin{multline}
\delta=O\left(\sqrt[3]{\epsilon_{D'}\log D'}\right)=O\left(D'^{-\frac{1-\alpha}{3\alpha}}\sqrt[3]{\log D'}\right)\\
\implies D'=\tilde O\left(\delta^{-\frac{3\alpha}{1-\alpha}}\right).
\end{multline}
Finally,
\begin{multline}
D=O\left(D'\sqrt[3]{\frac{\log^2D'}{\epsilon_{D'}}}\right)\implies D\delta=O(D'\log D')\\
\implies D=\tilde O\left(\delta^{-1-\frac{3\alpha}{1-\alpha}}\right).
\end{multline}
\end{IEEEproof}

\begin{IEEEproof} [Proof of Theorem \ref{t2}]
Lemma \ref{gap} yields
\begin{equation}
\log D'=\tilde O\left(\frac{1}{\Delta}+\sqrt[4]{\frac{1}{\Delta}\log^3\frac{1}{\epsilon_{D'}}}\right).
\end{equation}
Lemma \ref{trunc} gives
\begin{multline}
\delta=O\left(\sqrt[3]{\epsilon_{D'}\log D'}\right)=\tilde O(\sqrt[3]{\epsilon_{D'}})\\
\implies D'=e^{\tilde O\left(\frac{1}{\Delta}+\sqrt[4]{\frac{1}{\Delta}\log^3\frac{1}{\epsilon_{D'}}}\right)}=e^{\tilde O\left(\sqrt[4]{\frac{1}{\Delta}\log^3\frac{1}{\delta}}\right)}.
\end{multline}
Finally,
\begin{multline}
D=O\left(D'\sqrt[3]{\frac{\log^2D'}{\epsilon_{D'}}}\right)\implies D\delta=O(D'\log D')\\
\implies D=e^{\tilde O\left(\sqrt[4]{\frac{1}{\Delta}\log^3\frac{1}{\delta}}\right)}/\delta.
\end{multline}
\end{IEEEproof}

\section{Results in two dimensions}

We now extend Theorem \ref{t1} to 2D. Consider a square lattice of $N\times N$ qudits.

\begin{theorem} \label{2d}
Suppose that $|\psi\rangle$ satisfies an area law for the R\'enyi entanglement entropy $R_\alpha$ with $\alpha<1$, in the sense that the entropy between any rectangular region $X$ and its complement $\overline X$ is $O(|\partial X|)$, where $\partial X$ is the boundary of $X$. Then for any $\delta>0$, there exists a PEPS $|\phi\rangle$ with bond dimension
\begin{equation}
D=e^{O(1/\delta)}
\end{equation}
such that $|\langle\psi|\hat O|\psi\rangle-\langle\phi|\hat O|\phi\rangle|\le\delta$ for any local observable $\hat O$ with $\|\hat O\|\le1$.
\end{theorem}

\begin{figure}
\resizebox{\columnwidth}{!}{
\begin{tikzpicture}
\draw[help lines](-.5,-.5)grid(8.5,6.5);
\draw[ultra thick](1,1)--(1,5)--(2,5)--(2,1)--(3,1)--(3,5)--(4,5)--(4,1)--(5,1)--(5,5)--(6,5);
\draw[dashed,thick](2/3,2/3)rectangle(5+1/3,5+1/3);
\draw[dashed,thick](1/3,1/3)rectangle(7.5,5+2/3);
\begin{scope}[font=\small, inner sep=3]
\node[above right]at(1,1){1};
\node[right]at(1,2){2};
\node[right]at(1,3){3};
\node[right]at(1,4){4};
\node[below right]at(1,5){5};
\node[below right]at(2,5){6};
\node[right]at(2,4){7};
\node[left]at(5,4){24};
\node[below left]at(5,5){25};
\node[below left]at(6,5){26};
\end{scope}
\node[fill=white]at(5.5,2.5){$Y_{i,j}$};
\node[fill=white]at(7.5,2.5){$Y'_{i,j}$};
\draw[decorate,decoration={brace,amplitude=5}](0,1/2)--(0,5.5);
\node at(-.4,3){$m$};
\draw[decorate,decoration={brace,amplitude=5}](2/3,6)--(5+1/3,6);
\node at(3,6.4){$m$};
\draw[decorate,decoration={brace,amplitude=5}](5+2/3,6)--(7+1/3,6);
\node at(6.5,6.4){$b$};
\end{tikzpicture}
}
\caption{Construction of the state $|u_{i,j}\rangle$ in the proof of Theorem \ref{2d}, illustrated with parameters $m=5$ and $b=2$. The background grid represents the lattice. The inner dashed square and outer dashed rectangle enclose regions $Y_{i,j}$ and $Y'_{i,j}$, respectively. The numbers along the thick solid line show the order of the qudits. The path starts at site $(i,j)$ (labeled ``$1$'') and completely traverses $Y_{i,j}$ before entering $\overline{Y_{i,j}}$.}
\label{lattice}
\end{figure}

\begin{IEEEproof}
We index the qudits by $\{0,1,\ldots,N-1\}^2$, where the coordinates are taken modulo $N$ (periodic boundary conditions). See Fig.~\ref{lattice} for an illustration of the regions and the site ordering defined below. For any site $(i,j)$, let
\begin{gather}
Y_{i,j}=\{i,i+1,\ldots,i+m-1\}\times\{j,j+1,\ldots,j+m-1\},\\
\begin{split}
Y'_{i,j}=\{i,i+1,\ldots,i+m+b-1\}\\
\times\{j,j+1,\ldots,j+m-1\}
\end{split}
\end{gather}
be square and rectangular regions of sizes $m\times m$ and $(m+b)\times m$, respectively, satisfying $Y_{i,j}\subset Y'_{i,j}$. We express $|\psi\rangle$ exactly as an MPS in canonical form with bond dimension $\le d^{\lfloor N^2/2\rfloor}$, where the qudits are ordered as follows. We start at $(i,j)$ and then follow the thick solid line, which completely traverses $Y_{i,j}$ before reaching the first qudit outside $Y_{i,j}$. The remaining $N^2-m^2-1$ qudits in $\overline{Y_{i,j}}$ are ordered arbitrarily. Let $|u'_{i,j}\rangle$ be the unnormalized state obtained by truncating the first $m^2$ bonds (marked by the thick solid line) to dimension $D'$, and let $|u_{i,j}\rangle = |u'_{i,j}\rangle / \||u'_{i,j}\rangle\|$ be the normalized state. By Lemma \ref{area} and the area law for the R\'enyi entanglement entropy $R_\alpha(\alpha<1)$, we have
\begin{equation}
\left|\langle\psi|\hat O|\psi\rangle-\langle u_{i,j}|\hat O|u_{i,j}\rangle\right|\le\delta/2
\end{equation}
for any observable $\hat O$ with $\|\hat O\|\le1$, by choosing
\begin{equation}
D'=e^{O(m)}\poly(1/\delta).
\end{equation}
We may assume that
\begin{equation} \label{eq:c}
N\ge C/\delta
\end{equation}
for a sufficiently large constant $C>0$, since otherwise letting $|\phi\rangle=|u_{i,j}\rangle$ with $m=N$ and $b=0$ completes the proof.

Let $\mathcal H_X=\mathbb C^{d^{|X|}}$ denote the Hilbert space of qudits in region $X$, and let $I_X$ be the identity operator on $\mathcal H_X$. Set
\begin{equation}
b=\left\lceil\frac{\log_dD'}{m}\right\rceil=O\left(1+\frac{1}{m}\log\frac{1}{\delta}\right)
\end{equation}
so that $\dim\mathcal H_{Y'_{i,j}\setminus Y_{i,j}}\ge D'$. Since the Schmidt rank of $|u_{i,j}\rangle$ across the cut separating $Y_{i,j}$ from $\overline{Y_{i,j}}$ is $\le D'$ by construction, there exists an isometry $U_{i,j}:\mathcal H_{Y'_{i,j}\setminus Y_{i,j}}\to\mathcal H_{\overline{Y_{i,j}}}$ such that
\begin{equation}
|u_{i,j}\rangle=I_{Y_{i,j}}\otimes U_{i,j}|v_{i,j}\rangle
\end{equation}
for some state $|v_{i,j}\rangle\in\mathcal H_{Y'_{i,j}}$. Consequently, $\langle u_{i,j}|\hat O|u_{i,j}\rangle=\langle v_{i,j}|\hat O|v_{i,j}\rangle$ for any observable $\hat O$ on $\mathcal H_{Y_{i,j}}$. Let $V_{i,j}$ be some unitary on $\mathcal H_{Y'_{i,j}\setminus{Y_{i,j}}}$. Assume without loss of generality that $N$ is a common multiple of $m$ and $m+b$. We define
\begin{gather}
|\phi\rangle=\frac{1}{\sqrt{(m+b)m}}\sum_{i=1}^{m+b}\sum_{j=1}^{m}|w_{i,j}\rangle,\\
|w_{i,j}\rangle=\bigotimes_{k=0}^{\frac N{m+b}-1}\bigotimes_{l=0}^{\frac Nm-1}(I_{Y_{x,y}}\otimes V_{x,y})|v_{x,y}\rangle,
\end{gather}
where $x=i+k(m+b)$ and $y=j+lm$.

Analogous to (\ref{error3}), there exists a specific set of unitaries $\{V_{i,j}\}$ such that
\begin{equation}
\left|\langle\psi|\hat O|\psi\rangle-\langle\phi|\hat O|\phi\rangle\right|\le\delta/2+O(b/m)\le\delta
\end{equation}
for any local observable $\hat O$ with $\|\hat O\|\le1$, where we set $m=O(b/\delta)=O(1/\delta)$. By construction, each $|w_{i,j}\rangle$ is an MPS and thus a PEPS with bond dimension $e^{O(m)}\poly(1/\delta)$. Therefore, $|\phi\rangle$ is a PEPS with bond dimension
\begin{equation}
D=e^{O(m)}\poly(1/\delta)(m+b)m=e^{O(1/\delta)}.
\end{equation}
\end{IEEEproof}

\section{Logarithmic corrections to the area law}

Analogous results are obtained in both 1D and 2D for states with logarithmic corrections to the area law, such as many critical ground states \cite{HLW94, VLRK03, LRV04, CC04, CC09, Wol06, GK06}.

\begin{corollary} \label{c1}
In 2D, suppose that $|\psi\rangle$ satisfies an area law with logarithmic corrections for the R\'enyi entanglement entropy $R_\alpha$ with $\alpha<1$, in the sense that the entropy between any rectangular region $X$ and its complement $\overline X$ is $\tilde O(|\partial X|)$. Then for any $\delta>0$, there exists a PEPS $|\phi\rangle$ with bond dimension
\begin{equation}
D=e^{\tilde O(1/\delta)}
\end{equation}
such that $|\langle\psi|\hat O|\psi\rangle-\langle\phi|\hat O|\phi\rangle|\le\delta$ for any local observable $\hat O$ with $\|\hat O\|\le1$.
\end{corollary}

\begin{IEEEproof}
This follows from the proof of Theorem \ref{2d} by modifying the parameters to
\begin{equation}
m=\tilde O(1/\delta),\quad D'=e^{\tilde O(m)}\poly(1/\delta)
\end{equation}
and adjusting (\ref{eq:c}) accordingly.
\end{IEEEproof}

\begin{corollary} \label{c2}
In 1D, suppose that $|\psi\rangle$ satisfies an area law with logarithmic corrections for the R\'enyi entanglement entropy $R_\alpha$ with $\alpha<1$, in the sense that the entropy between any contiguous region $X$ and its complement $\overline X$ is $\le c_\alpha\ln|X|+O(1)$. Then for any $\delta>0$, there exists an MPS $|\phi\rangle$ with bond dimension
\begin{equation}
D=\tilde O\left(\delta^{-1-\frac{3\alpha}{1-\alpha}-c_\alpha}\right)
\end{equation}
such that $|\langle\psi|\hat O|\psi\rangle-\langle\phi|\hat O|\phi\rangle|\le\delta$ for any local observable $\hat O$ with $\|\hat O\|\le1$.
\end{corollary}

\begin{IEEEproof}
We adapt the proof of Theorem \ref{2d} from 2D to 1D by taking $Y$ and $Y'$ to be contiguous regions of $m$ and $m+b$ qudits, respectively. We modify the parameters to
\begin{equation}
m=\tilde O(1/\delta),\quad D'=\tilde O\left(\delta^{-\frac{3\alpha}{1-\alpha}-c_\alpha}\right)
\end{equation}
and adjust (\ref{eq:c}) accordingly.
\end{IEEEproof}

\section*{Acknowledgment}

During the revision of this paper, I used Gemini 2.5 Pro Preview, Gemini 3 Pro Preview, Gemini 3.1 Pro Preview, and OpenAI o4-mini to improve the presentation. I also used Gemini 3 Pro Preview and Gemini 3.1 Pro Preview to help verify the results. I would like to thank Google DeepMind for providing Gemini API credits.

\bibliographystyle{IEEEtran}
\bibliography{jrnl}

% Generated by IEEEtran.bst, version: 1.14 (2015/08/26)
\begin{thebibliography}{10}
\providecommand{\url}[1]{#1}
\csname url@samestyle\endcsname
\providecommand{\newblock}{\relax}
\providecommand{\bibinfo}[2]{#2}
\providecommand{\BIBentrySTDinterwordspacing}{\spaceskip=0pt\relax}
\providecommand{\BIBentryALTinterwordstretchfactor}{4}
\providecommand{\BIBentryALTinterwordspacing}{\spaceskip=\fontdimen2\font plus
\BIBentryALTinterwordstretchfactor\fontdimen3\font minus
  \fontdimen4\font\relax}
\providecommand{\BIBforeignlanguage}[2]{{%
\expandafter\ifx\csname l@#1\endcsname\relax
\typeout{** WARNING: IEEEtran.bst: No hyphenation pattern has been}%
\typeout{** loaded for the language `#1'. Using the pattern for}%
\typeout{** the default language instead.}%
\else
\language=\csname l@#1\endcsname
\fi
#2}}
\providecommand{\BIBdecl}{\relax}
\BIBdecl

\bibitem{DB19}
A.~M. Dalzell and F.~G. S.~L. Brand{\~{a}}o, ``Locally accurate {MPS}
  approximations for ground states of one-dimensional gapped local
  {H}amiltonians,'' \emph{Quantum}, vol.~3, p. 187, 2019.

\bibitem{Hua19}
Y.~Huang, ``Matrix product state approximations: {B}ringing theory closer to
  practice,'' \emph{Quantum Views}, vol.~3, p.~26, 2019.

\bibitem{PVWC07}
D.~Perez-Garcia, F.~Verstraete, M.~M. Wolf, and J.~I. Cirac, ``Matrix product
  state representations,'' \emph{Quantum Information and Computation}, vol.~7,
  no. 5-6, pp. 401--430, 2007.

\bibitem{FNW92}
M.~Fannes, B.~Nachtergaele, and R.~F. Werner, ``Finitely correlated states on
  quantum spin chains,'' \emph{Communications in Mathematical Physics}, vol.
  144, no.~3, pp. 443--490, 1992.

\bibitem{VC04}
F.~Verstraete and J.~I. Cirac, ``Renormalization algorithms for quantum-many
  body systems in two and higher dimensions,'' arXiv:cond-mat/0407066.

\bibitem{Sch11}
U.~Schollw{\"o}ck, ``The density-matrix renormalization group in the age of
  matrix product states,'' \emph{Annals of Physics}, vol. 326, no.~1, pp.
  96--192, 2011.

\bibitem{Whi92}
S.~R. White, ``Density matrix formulation for quantum renormalization groups,''
  \emph{Physical Review Letters}, vol.~69, no.~19, pp. 2863--2866, 1992.

\bibitem{Whi93}
------, ``Density-matrix algorithms for quantum renormalization groups,''
  \emph{Physical Review B}, vol.~48, no.~14, pp. 10\,345--10\,356, 1993.

\bibitem{Has07}
M.~B. Hastings, ``An area law for one-dimensional quantum systems,''
  \emph{Journal of Statistical Mechanics: Theory and Experiment}, vol. 2007,
  no.~08, p. P08024, 2007.

\bibitem{AKLV13}
I.~Arad, A.~Kitaev, Z.~Landau, and U.~Vazirani, ``An area law and
  sub-exponential algorithm for {1D} systems,'' arXiv:1301.1162.

\bibitem{Hua14a}
Y.~Huang, ``Area law in one dimension: {D}egenerate ground states and {R}enyi
  entanglement entropy,'' arXiv:1403.0327.

\bibitem{LVV15}
Z.~Landau, U.~Vazirani, and T.~Vidick, ``A polynomial time algorithm for the
  ground state of one-dimensional gapped local {H}amiltonians,'' \emph{Nature
  Physics}, vol.~11, no.~7, pp. 566--569, 2015.

\bibitem{Hua14b}
Y.~Huang, ``A polynomial-time algorithm for the ground state of one-dimensional
  gapped {H}amiltonians,'' arXiv:1406.6355.

\bibitem{CF16}
C.~T. Chubb and S.~T. Flammia, ``Computing the degenerate ground space of
  gapped spin chains in polynomial time,'' \emph{Chicago Journal of Theoretical
  Computer Science}, vol. 2016, p.~9, 2016.

\bibitem{Hua15c}
Y.~Huang, ``A simple efficient algorithm in frustration-free one-dimensional
  gapped systems,'' arXiv:1510.01303.

\bibitem{ALVV17}
I.~Arad, Z.~Landau, U.~Vazirani, and T.~Vidick, ``Rigorous {RG} algorithms and
  area laws for low energy eigenstates in {1D},'' \emph{Communications in
  Mathematical Physics}, vol. 356, no.~1, pp. 65--105, 2017.

\bibitem{SWVC08}
N.~Schuch, M.~M. Wolf, F.~Verstraete, and J.~I. Cirac, ``Entropy scaling and
  simulability by matrix product states,'' \emph{Physical Review Letters}, vol.
  100, no.~3, p. 030504, 2008.

\bibitem{McC08}
I.~P. McCulloch, ``Infinite size density matrix renormalization group,
  revisited,'' arXiv:0804.2509.

\bibitem{Vid07}
G.~Vidal, ``Classical simulation of infinite-size quantum lattice systems in
  one spatial dimension,'' \emph{Physical Review Letters}, vol.~98, no.~7, p.
  070201, 2007.

\bibitem{Hua15a}
Y.~Huang, ``Computing energy density in one dimension,'' arXiv:1505.00772.

\bibitem{SV17}
N.~Schuch and F.~Verstraete, ``Matrix product state approximations for infinite
  systems,'' arXiv:1711.06559.

\bibitem{VGC04}
F.~Verstraete, J.~J. Garc\'{\i}a-Ripoll, and J.~I. Cirac, ``Matrix product
  density operators: Simulation of finite-temperature and dissipative
  systems,'' \emph{Physical Review Letters}, vol.~93, no.~20, p. 207204, 2004.

\bibitem{ZV04}
M.~Zwolak and G.~Vidal, ``Mixed-state dynamics in one-dimensional quantum
  lattice systems: A time-dependent superoperator renormalization algorithm,''
  \emph{Physical Review Letters}, vol.~93, no.~20, p. 207205, 2004.

\bibitem{KGE14}
M.~Kliesch, D.~Gross, and J.~Eisert, ``Matrix-product operators and states:
  {NP}-hardness and undecidability,'' \emph{Physical Review Letters}, vol. 113,
  no.~16, p. 160503, 2014.

\bibitem{ECP10}
J.~Eisert, M.~Cramer, and M.~B. Plenio, ``Colloquium: {A}rea laws for the
  entanglement entropy,'' \emph{Reviews of Modern Physics}, vol.~82, no.~1, pp.
  277--306, 2010.

\bibitem{Osb12}
T.~J. Osborne, ``Hamiltonian complexity,'' \emph{Reports on Progress in
  Physics}, vol.~75, no.~2, p. 022001, 2012.

\bibitem{GHLS15}
S.~Gharibian, Y.~Huang, Z.~Landau, and S.~W. Shin, ``Quantum {H}amiltonian
  complexity,'' \emph{Foundations and Trends in Theoretical Computer Science},
  vol.~10, no.~3, pp. 159--282, 2015.

\bibitem{Hua15b}
Y.~Huang, ``Classical simulation of quantum many-body systems,'' Ph.D.
  dissertation, University of California, Berkeley, 2015.

\bibitem{GE16}
Y.~Ge and J.~Eisert, ``Area laws and efficient descriptions of quantum
  many-body states,'' \emph{New Journal of Physics}, vol.~18, no.~8, p. 083026,
  2016.

\bibitem{Hua20ISIT}
Y.~Huang, ``2{D} {L}ocal {H}amiltonian with area laws is {QMA}-complete,'' in
  \emph{2020 IEEE International Symposium on Information Theory}, 2020, pp.
  1927--1932.

\bibitem{Hua21JCP}
------, ``Two-dimensional local {H}amiltonian problem with area laws is
  {QMA}-complete,'' \emph{Journal of Computational Physics}, vol. 443, p.
  110534, 2021.

\bibitem{HLW94}
C.~Holzhey, F.~Larsen, and F.~Wilczek, ``Geometric and renormalized entropy in
  conformal field theory,'' \emph{Nuclear Physics B}, vol. 424, no.~3, pp.
  443--467, 1994.

\bibitem{VLRK03}
G.~Vidal, J.~I. Latorre, E.~Rico, and A.~Kitaev, ``Entanglement in quantum
  critical phenomena,'' \emph{Physical Review Letters}, vol.~90, no.~22, p.
  227902, 2003.

\bibitem{LRV04}
J.~I. Latorre, E.~Rico, and G.~Vidal, ``Ground state entanglement in quantum
  spin chains,'' \emph{Quantum Information and Computation}, vol.~4, no.~1, pp.
  48--92, 2004.

\bibitem{CC04}
P.~Calabrese and J.~Cardy, ``Entanglement entropy and quantum field theory,''
  \emph{Journal of Statistical Mechanics: Theory and Experiment}, vol. 2004,
  no.~06, p. P06002, 2004.

\bibitem{CC09}
------, ``Entanglement entropy and conformal field theory,'' \emph{Journal of
  Physics A: Mathematical and Theoretical}, vol.~42, no.~50, p. 504005, 2009.

\bibitem{Wol06}
M.~M. Wolf, ``Violation of the entropic area law for fermions,'' \emph{Physical
  Review Letters}, vol.~96, no.~1, p. 010404, 2006.

\bibitem{GK06}
D.~Gioev and I.~Klich, ``Entanglement entropy of fermions in any dimension and
  the {W}idom conjecture,'' \emph{Physical Review Letters}, vol.~96, no.~10, p.
  100503, 2006.

\bibitem{Hua20}
Y.~Huang, ``Computing local properties in the trivial phase,''
  arXiv:2001.10763.

\bibitem{VC06}
F.~Verstraete and J.~I. Cirac, ``Matrix product states represent ground states
  faithfully,'' \emph{Physical Review B}, vol.~73, no.~9, p. 094423, 2006.

\end{thebibliography}

\begin{IEEEbiographynophoto}{Yichen Huang}
received the B.Sc. degree in mathematics and physics from Tsinghua University in 2010, and the Ph.D. degree in physics from the University of California at Berkeley, in 2015. He was an IQIM Postdoctoral Scholar with the California Institute of Technology from 2015 to 2018, an Associate Researcher with Microsoft Research AI from 2018 to 2019, a Senior Postdoctoral Associate with the Center for Theoretical Physics, Massachusetts Institute of Technology from 2019 to 2022, and a Postdoctoral Fellow with the Department of Physics, Harvard University from 2022 to 2024. His research interests include artificial intelligence, quantum information theory, and condensed matter theory.
\end{IEEEbiographynophoto}

\end{document}